\documentclass[12pt]{article}
\usepackage{amsthm,amsfonts,natbib,mathrsfs,amssymb,fancybox,amsmath}
\usepackage{diagxy}
\usepackage{fullpage}

\newtheorem*{swt}{Specialized Representation Wigner Theorem}
\newtheorem*{lemma}{Lemma}
\newtheorem*{cor}{Corollary 1}
\newtheorem*{cor2}{Corollary 2}
\newtheorem*{wt}{Wigner's Theorem}
\newtheorem*{gwt}{Representation Wigner Theorem}
\newtheorem*{gns}{GNS Theorem}
\theoremstyle{definition}
\newtheorem*{defn}{Definition}

\newcommand{\norm}[1]{\|#1\|}

\def\a{\alpha}
\def\Om{\Omega}

\def\om{\omega}
\newcommand{\Hil}{\mathcal{H}}
\newcommand{\Alg}{\mathfrak{A}}

\def\2#1{{\mathfrak #1}}
\def\3#1{{\mathcal #1}}
\def\7#1{{\mathbb #1}}
\newcommand{\al}{\alpha}

\begin{document}

\title{How is spontaneous symmetry breaking possible?}
\author{David John Baker\thanks{djbaker@umich.edu} \space and Hans
Halvorson\thanks{hhalvors@princeton.edu}}
\date{}

\maketitle
\begin{abstract}
  We pose and resolve a seeming paradox about spontaneous symmetry
  breaking in the quantum theory of infinite systems.  For a symmetry
  to be spontaneously broken, it must not be implementable by a
  unitary operator.  But Wigner's theorem guarantees that every
  symmetry is implemented by a unitary operator that preserves
  transition probabilities between pure states.  We show how it is
  possible for a unitary operator of this sort to connect the folia of
  unitarily inequivalent representations.  This result undermines
  interpretations of quantum theory that hold unitary equivalence to
  be necessary for physical equivalence.
\end{abstract}

\section{Introduction}

The precise mathematical definition of spontaneous symmetry
breaking (SSB) in quantum theory is somewhat up for grabs.  But
all hands agree that, in the case of infinitely many degrees of
freedom, unitarily inequivalent representations are needed.

In physics more generally, SSB occurs when a ground state is not
invariant under a symmetry of the laws.  This means that a
symmetry transformation will take a ground state to another
(mathematically distinct) ground state.  But in quantum field
theory, an (irreducible) Hilbert space representation of the
commutation relations can include only a single vacuum state.  So
a spontaneously broken symmetry must map between different
unitarily inequivalent representations.

Thus for SSB to occur in infinite quantum theory, the broken
symmetry must not be implemented by a unitary operator. This is
generally agreed to be a necessary condition \citep{GEandCLssb}
and is sometimes taken to be both necessary and sufficient
\citep{JEroughssb,FSsymbreaking}.

Put this way, it can be difficult to see how a symmetry can
possibly be spontaneously broken.  The difficulty arises from an
apparent conflict with Wigner's unitary-antiunitary theorem, a
foundational result that applies to all quantum theories. Since a
symmetry ought to preserve all the empirical predictions of a
quantum state, it must not change the transition probabilities
between pure states, which are represented by the inner products
between vectors in a Hilbert space representation.  Wigner's
theorem shows that any mapping that preserves these probabilities
for all vector states in a Hilbert space must be a unitary
mapping.

All symmetries preserve transition probabilities in this way.  Besides
being physically intuitive, this can be proven rigorously even in
paradigm cases of SSB.  But this seems to lead to paradox.  We know
SSB is possible in infinite quantum theory --- there are well-known
examples.  But we also know that any symmetry, even a broken one, must
satisfy the premises of Wigner's theorem.  Since SSB requires unitary
inequivalence, these two facts appear inconsistent.

This inconsistency must be only apparent.  Our task is to explain
why.  We'll begin by explaining some general features that apply
in all cases of quantum SSB.
We will show that in such cases, Wigner's theorem applies.  The
seeming paradox therefore threatens.  To resolve it, we'll show
that the existence of a unitary symmetry in Wigner's sense does
not entail the unitary equivalence of the Hilbert space
representations it connects.  There remains a sense in which the
symmetry is not (strictly speaking) unitarily implementable.

\section{Quantum SSB}

We begin by recalling some general properties of quantum theories on
the algebraic approach.  At the broadest level of generality, a
quantum theory is described by a $C^*$-algebra $\Alg$ obeying the
canonical commutation or anticommutation relations (CCRs or CARs
respectively) in their bounded form. This is either an algebra of
observables or (as in the present case) a field algebra.  The
self-adjoint operators in $\Alg$ stand for physical quantities and are
often called observables. The states of the algebra are the possible
assignments of expectation values to the operators in $\Alg$.  These
are given by normed linear functionals $\omega: \Alg \rightarrow
\mathbb{C}.$ The expectation value of $A$ in state $\om$ is written
$\om(A)$.

A clear connection exists between this algebraic formalism and the
better-known Hilbert space formalism. The abstract algebraic states
and observables can be concretely realized by a Hilbert space
representation of $\Alg$ (also called a representation of the
CCRs/CARs).  Such a representation is a mapping $\pi$ from $\Alg$ into
the algebra of bounded operators $B(\Hil)$ on a Hilbert space $\Hil$.
The representation map is not usually a bijection; Hilbert space
representations will include more operators, and in particular more
observables, than the $C^*$-algebra.  Some of the states $\omega$ of
$\Alg$ will be representable by density operators on $\Hil$ that agree
with their expectation values for all $A$ in $\Alg$.  Collectively
these states are called the \emph{folium} of the representation $\pi$.

Every algebraic state $\omega$ has a unique ``home''
representation in which it is given by a cyclic vector.  This is
established by the
\begin{gns} For each state $\om$ of $\Alg$, there is a representation
  $\pi$ of $\Alg$ on a Hilbert space $\Hil$, and a vector $\Om \in
  \Hil$ such that $\om (A)=\langle \Om ,\pi (A) \Om \rangle$, for all
  $A\in \Alg$, and the vectors $\{ \pi (A)\Om : A\in \Alg\}$ are dense
  in $\Hil$.  (Call any representation meeting these criteria a \emph{GNS representation.})  The GNS representation is unique in the sense that for any other representation $(\3H ',\pi ',\Om ')$ satisfying the previous
  two conditions, there is a unique unitary operator $U:\Hil \to \Hil
  '$ such that $U\Om =\Om '$ and $U\pi (A)=\pi '(A)U$, for all $A$ in
  $\Alg$  \emph{\citep[see][278--279]{kr}}.\end{gns}The definition of
  ``same representation'' presumed in this statement of uniqueness is called \emph{unitary
  equivalence}.  We call representations $\pi$ and $\pi'$
\emph{unitarily inequivalent}, and treat them as
distinct,\footnote{Although we always treat inequivalent
  representations as formally or mathematically distinct, note that
  they may not always be physically inequivalent.  As we shall see,
  they are sometimes related by symmetries, which are normally assumed
  to preserve all the physical facts.} if there is no unitary operator
$U$ between their Hilbert spaces which relates the representations by
\begin{equation}\label{EQuequiv}
    U\pi(A)=\pi' (A)U.
\end{equation}
When equation (\ref{EQuequiv}) does hold, we say that the unitary $U$ \emph{intertwines} the representations $\pi$ and $\pi'$.



A useful source on the representation of symmetry in this
framework is \citet{JRandGRbasicaqft}.  They posit (very
reasonably) that any symmetry of a quantum system must at a
minimum consist of two bijections, $\a$ from the algebra of
physical quantities $\Alg$ onto itself and $\a'$ from the space of
states of $\Alg$ onto itself.  These must preserve all expectation
values, so that
\begin{equation}\label{EQsym}
    \a'(\om(\a(A)))=\om(A).
\end{equation}
They then show that any such $\a$ is a $*$-automorphism of $\Alg$, a
bijection $\a:\Alg \to \Alg$ which preserves its algebraic structure
and commutes with the adjoint mapping $(\cdot)^*$.  Furthermore,
$\alpha'$ can be defined in terms of this $*$-automorphism as it acts
on states. Clearly $\a'(\omega)$ is given by $\om \circ \a^{-1}=\omega(\a^{-1}(A))$. We
therefore have a justification, from physical principles, of the
oft-cited fact that symmetries in quantum theory are given by
$*$-automorphisms.

Clearly if $\a$ is a symmetry and $\pi(A)$ is a representation of $\Alg$ on a Hilbert space $\Hil$, $\pi \circ \a(A)=\pi(\alpha(A))$ is also a representation of $\Alg$ on $\Hil$.   In this case $\a$ will act as a bijective mapping from $\pi$ to $\pi \circ \a$.  We call $\alpha$ \emph{unitarily implementable in the representation $\pi$} when there is a unitary mapping $U:\3H\to\3H$ such that
\begin{equation}\label{EQimplem}
  \pi'(A)=\pi(\alpha(A))= U \pi(A) U^* .
\end{equation}
This means the symmetry $\alpha$ is unitarily implementable in $\pi$ iff $\pi$ and
$\pi \circ \a$ are unitarily equivalent representations of $\Alg$.

This is where spontaneous symmetry breaking comes in.  In general, a state $\om$ breaks the symmetry $\a$ only if $\a$ is not unitarily implementable in $\om$'s GNS representation.\footnote{On the approach shared by Strocchi and Earman, this is both a necessary and sufficient condition.}  When this occurs, $\om$'s GNS representation and the GNS representation of the symmetry-transformed state $\a'(\om)=\om \circ \a^{-1}$ will be unitarily inequivalent.

In one example, the CAR algebra (so-called because it obeys the
canonical anticommutation relations) is the field algebra for a system
of interacting spin-1/2 systems.  We may use the infinite version of
the algebra to represent an infinitely long chain of spins confined to
a one-dimensional lattice, as in the Heisenberg model of a
ferromagnet.  This infinite CAR algebra possesses a non-unitarily implementable
automorphism which represents a symmetry of the ferromagnet: namely, a
180-degree rotation which flips all of the spins in the chain.  The
rotation is therefore a spontaneously broken symmetry. See
\citet{LRferro} for a detailed study of this case; we present only a
few general features it shares in common with other examples of SSB.

The lowest-energy states available to an infinite spin chain are
ones in which all of the spins align in the same direction. The
Heisenberg ferromagnet has two such ground states: $\om$, in which
all of the spins point along $+x$ (where $x$ is the axis of the
one-dimensional chain) and $\om'$, in which they all point along
$-x$.  These states each define a GNS representation ($\pi,\pi'$
respectively) on Hilbert spaces $\Hil$ and $\Hil'$.   If $\alpha$
is the automorphism of the CAR algebra representing a 180-degree
rotation along an axis perpendicular to $x$, $\om'=\alpha(\om)$.
But since $\alpha$ is spontaneously broken, $\pi$ must be
unitarily inequivalent to $\pi'$.  This is where the seeming
paradox comes in.

\section{Wigner's theorem and the paradox}

The expectation values of observables aren't the only important
quantities in quantum physics.  We should also expect a symmetry
to preserve the transition probabilities between (pure) states of
any quantum theory.  In the Hilbert space formalism these are
given by inner products: $\langle \psi,\psi' \rangle$ represents the
likelihood of a spontaneous transition from  vector state
$\psi$ to $\psi'$.\footnote{Only on collapse
interpretations do such transitions actually occur, of course. But
we should nevertheless expect other interpretations to retain the
statistical predictions codified in transition probabilities.}

It seems obvious that no symmetry worth its salt will alter any of
the transition probabilities.  This led Wigner to conclude that
any symmetry worth its salt is given by a unitary
operator.\footnote{Or by an antiunitary operator; for present
purposes we ignore the difference.}  For the following can be
proven \citep[see][]{VBwignersthm}:
\begin{wt} Any bijection from the unit rays (vectors states) of a Hilbert space $\Hil$
  to the unit rays of $\Hil'$ which preserves the inner product is
  given by a unitary mapping $W:\Hil\to\Hil'$.
\end{wt}
This is puzzling.  A spontaneously broken symmetry should still map
bijectively between the vector states of the two representations, and
it would be very strange if it failed to preserve transition
probabilities.  But we also know that broken symmetries are not
unitarily implementable.  This seems impossible, but since well-known
examples of SSB exist, paradox threatens.

We are not the first to notice this problem.  Considering the
$C^{*}$-algebra $B(\Hil)$ of all bounded operators on a Hilbert space,
Earman presents the problem as follows:
\begin{quote}
Obviously, if $U$ is unitary then [the map $A \rightarrow UAU^*$]
is an automorphism of $B(\Hil)$.  Conversely, any automorphism of
$B(\Hil)$ takes this form for a unitary $U$ -- ... $T$ is an inner
automorphism\footnote{An automorphism $\alpha$ is inner just in
case there is a unitary $V\in \mathfrak{A}$ such that $\alpha
(A)=VAV^*$ for all $A\in \mathfrak{A}$}
 of the $C^*$-algebra $B(\Hil)$.  How then can a symmetry
in the guise of an automorphism of a $C^*$-algebra $\Alg$ fail to be
an unbroken or Wigner symmetry [one which preserves transition
probabilities]?
 \citep[341]{JEroughssb}
\end{quote}
Earman provides an answer, but not a satisfactory one.  He notes
correctly that in cases of SSB, the automorphism that represents
the broken symmetry is not an automorphism of all of $B(\Hil)$.
Instead it is defined only on the field algebra or algebra of
observables (the CAR algebra in the ferromagnet example).  In a
Hilbert space representation this distinguished algebra is given
by a much smaller subalgebra of $B(\Hil)$.  So since the symmetry
is not an automorphism of $B(\Hil)$, there is no reason to assume
it should be unitary in the first place.

This only appears to resolve the problem because Earman has presented
it in a limited form.  He considers the symmetry solely as it acts on
operators, but the paradox can only be properly motivated by
considering the symmetry's action on states.  Whether $\alpha$ is an
automorphism of $B(\Hil)$ or not, $\a'$ is still a bijection between
the states of representations $\Hil$ and $\Hil'$.  If it's not
unitary, it must fail to meet one of the premises of Wigner's theorem.

The problem is that the premises of Wigner's theorem are satisfied by
\emph{all} symmetries.  Consider $\a'$ as it maps between the GNS
representations of two ground states.  Since a symmetry always takes
pure states to pure states, and the vector states of a GNS
representation are all and only its pure states, $\a'$ is indeed a
bijection of vectors states (unit rays).  Furthermore, as
\citet[335]{JRandGRbasicaqft} prove, $\a'$ must preserve all
transition probabilities.  Wigner's theorem therefore applies.

Earman is not alone in struggling with this seeming paradox.  Other authors have suggested (contrary to the theorem of Roberts and Roepstorff) that perhaps spontaneously broken symmetries do fail to preserve transition probabilities.  For example, \citet[302 fn 111]{AAfields} claims that, because ``the preservation of probabilities in the Hilbert space formulation implies the existence of a unitary (or antiunitary) operator,'' no mapping between the folia of inequivalent representations -- presumably not even a symmetry transformation -- can preserve transition probabilities. Even \citet[483]{LRferro} makes a rare slip on this point:
\begin{quote} A state $\om$ breaks a dynamical symmetry $\a$ iff $\a$
  is not unitarily implementable on $\om$'s GNS representation.
  Unlike rotations on ground states of the finite spin chain,
  \emph{dynamical symmetries do not conserve transition probabilities}
  in the GNS representation of states that break them. (Emphasis
  ours) \end{quote} But to the contrary, by the very definition of a symmetry, all
symmetries --- whether broken or not --- preserve transition
probabilities.  Therefore, broken symmetries satisfy the premises of
Wigner's theorem.

So the answer to Earman's question --- how can a symmetry fail to be a
Wigner-type unitary symmetry? --- is that it cannot.  And yet, there
are unitarily non-implementable symmetries.  Paradox threatens.

\section{The resolution}

Since Wigner's theorem applies to all symmetries, a spontaneously
broken symmetry must in some sense give a unitary mapping between the
states of unitarily inequivalent representations.  So there must be
some wiggle room in the definition of unitary equivalence that makes
this possible.  To resolve the paradox, we must look again at the
definition and find the wiggle room.

In effect, we have two data points to work with.  First, as Roberts and Roepstorff prove, Wigner's theorem applies to all algebraic symmetries.  This means that any symmetry $\a'$, as it acts on states, must be implemented by some unitary operator.  Since the existence of this operator is guaranteed by Wigner's theorem, we'll call it the \emph{Wigner unitary} $W$.  Since for any state $\om$, $\a'(\om)=\om \circ \a^{-1}$, Wigner's theorem is telling us that $W$ must take the state vector that represents $\om$ to a vector that represents $\om \circ \a^{-1}$.

Our second data point is the fact that spontaneous symmetry breaking is possible.  This implies that some symmetries are not unitarily implementable, in the sense that they fail to satisfy Eq. (\ref{EQimplem}) for any unitary $U$.  When Eq. (\ref{EQimplem}) is not satisfied, the symmetry's action on operators (which takes each operator $A$ to $\a(A)$) cannot be implemented by a unitary operator.  Combining our two data points, this means that spontaneously broken symmetries are unitarily implementable as they act on states, but not as they act on operators.

This means that the unitary $W$ must map the vector states of $\Hil$ to those of $\Hil'$ without satisfying Eq.~\ref{EQuequiv} for $U=W$.  In such a case, the representations are unitarily inequivalent even though their Hilbert spaces are related by a unitary operator. There is nothing contradictory about this, since the existence of a unitary operator implies unitary equivalence only if the operator intertwines the representations.  Their respective representations $\pi,\pi'$ of the $C^*$-algebra may nonetheless be preserved by $W$, as long as $W$ does not map the representation $\pi$ pointwise to the representation $\pi'$.  So it may still be that
\begin{equation}
    W \pi(\Alg) = \pi' (\Alg) W  \label{foobar}
\end{equation}
although there are individual operators $A \in \Alg$ for which (\ref{EQuequiv}) does not hold.

In fact, an operator meeting these criteria exists whenever the states of two GNS representations are connected by a symmetry.  We will show in steps that in every such case a unitary $W$ exists which satisfies Eq. (\ref{foobar}), and which implements the symmetry as it acts on states without implementing it as it acts on operators.  First, we establish its existence:

\begin{gwt} Let $\langle \8H,\pi,\Om\rangle$ be a GNS representation
  for $\om$, and let $\langle \8H',\pi',\Om '\rangle$ be a GNS
  representation for $\om\circ\alpha^{-1}$.  Then there is a unique
  unitary operator $W=W_{\pi ,\pi'}:\8H\to \8H'$ such that $W\Om =\Om
  '$, and $W\pi (\alpha ^{-1}(A))= \pi'(A)W$ for all $A\in \2A$.  \emph{(Proof in Appendix 1.)}
\end{gwt}

Since $\a^{-1} (\Alg)=\Alg$ and $W$ intertwines $\pi \circ \a^{-1}$ and $\pi'$, Eq. (\ref{foobar}) follows.  This means that when $\pi$ and $\pi'$ are not unitarily equivalent, $W$ maps between these two representations without mapping $\pi$ pointwise to $\pi'$, which is what we expected.  In other words, $W$ acts as a bijection between the operators of these representations but does not, in general, implement the symmetry as it acts on operators.

We have established that a unitary mapping preserving transition probabilities can exist even between unitarily inequivalent representations.  Indeed, such a mapping always exists in cases of SSB.  This is not yet enough to ensure that the conclusion of Wigner's theorem is true (as it must be).  Wigner's theorem ensures, not just that such a mapping exists, but that every mapping which preserves transition probabilities must be unitary.  This includes every symmetry (whether spontaneously broken or not) as it acts on states.

So we must also show that a spontaneously broken symmetry can be given by a unitary operator in the sense just discussed, without being unitarily implemented --- that is, without satisfying Eq.~\ref{EQimplem}.  To establish this, we will show that $W$ itself implements the symmetry as it acts on states.

Keep in mind that any representation $\pi :\2A \to B(\8H )$ of a $C^*$-algebra gives rise to a map $\pi^*$ of unit vectors of $\8H$ into the state space of $\2A$.  In particular, $$\pi ^*(x)(A) = \langle x,\pi (A)x\rangle ,\qquad (A\in \2A) .$$  We now use this map to show that $W$ implements the symmetry $\a'$ as it acts between the states of representations $\pi$ and $\pi'$.

\begin{cor} Let $(\8H ,\pi ,\Om)$ be a GNS representation for $\om$.
  Then the Wigner unitary $W$ for $\al$ implements the action of $\al$
  on vectors in $\8H$.  That is, $(\pi ')^*(Wx)=\pi ^*(x)\circ \al
  ^{-1}$ for any unit vector $x$ in $\8H$.  \emph{(Proof in Appendix 1.)}\end{cor}

In other words, when we apply $W$ to the state vector $x \in \Hil$ which represents the algebraic state $\pi^*(x)$ in the GNS representation $\pi$, the result is the vector $Wx \in \Hil'$ which represents the state $\a'(\pi^*(x))=\pi^*(x) \circ \a^{-1}$ in the representation $\pi'$.  This is just what it means for $W$ to implement the symmetry as it acts on states.

Finally, we confirm that $W$ does not in general implement the symmetry $\a$ as it acts on operators. This is just to say that it does not in general intertwine the representations $\pi$ and $\pi'$.  In fact, we can show that $W$ intertwines these representations only if it is trivial:
\begin{cor2} If the Wigner unitary $W:\8H\to\8H'$ also induces a
  unitary equivalence between $\pi$ and $\pi'$, then $\al '\circ \pi
  ^*=\pi ^*$.  \label{trivialize} \emph{(Proof in Appendix 1.)} \end{cor2}
That is, every vector state in $\pi$'s Hilbert space is invariant under the symmetry $\a'$, and hence left unchanged by $W$.  This means $W$ must be the identity.

The astute reader may be puzzled by some of the properties we ascribe to the Wigner unitary.  In particular, we've shown that the Wigner unitary, which implements a symmetry as it acts on states, never implements that same symmetry as it acts on operators unless it is the trivial identity operator.  How, then, can a non-trivial symmetry be unitarily implemented on both states and operators (and hence unbroken)?  This sort of puzzle is best resolved by looking at concrete examples of Wigner unitaries in the case of both broken and unbroken symmetries, which examples we provide in Appendix 2.

We've shown that whenever two states are related by a symmetry, a unitary mapping exists between the Hilbert spaces of their GNS representations and has the properties we would expect.  This is the Wigner unitary.  Its existence vindicates Wigner's theorem, in that it shows how the theorem can be true even when spontaneous symmetry breaking prevents a symmetry from being unitarily implemented as it acts on operators.  The seeming paradox is no paradox at all.

\section{Foundational significance}

Besides the dissolution of a confusing paradox, are there foundational implications of this result?  We believe so.  To underscore the foundational importance of our resolution of the paradox, let's briefly explore how it bears on one vexed question in the philosophy of quantum field theory.  What are the necessary conditions for physical equivalence between field-theoretic states?  In AQFT, the representations of the field algebra separate the states into natural ``families:'' the folia of states given by density operators in each representation.  We may therefore ask what conditions must be met for two such families of states -- two folia -- to represent the same set of physical possibilities.

For two folia to be physically equivalent, they must at least be empirically equivalent.  The paradoxical line of reasoning suggests that unitary equivalence is necessary if we want to preserve transition probabilities.   Since a quantum theory's transition probabilities are part of its empirical content, it would seem to follow that the folia of unitarily inequivalent representations cannot predict the same empirical consequences -- making them physically inequivalent by the above reasoning.

This is why Arageorgis, while attempting ``to clarify the
connection between `intertranslatability' [a necessary condition for physical equivalence] and `unitary equivalence,'" writes in his seminal dissertation,
\begin{quote}
Intertranslatability requires a mapping between theoretical
descriptions that preserves the reports of empirical findings.
These are couched in terms of probabilities in quantum theory. And
as Wigner has taught us, the preservation of probabilities in the
Hilbert space formulation implies the existence of a unitary (or
antiunitary)\footnote{Recall that for purposes of this paper we
ignore the distinction between unitary and antiunitary.} operator.
\citep[302 fn 111]{AAfields}
\end{quote}
He takes this point to establish that folia must belong to unitarily equivalent representations if they are to count as physically equivalent.  But his argument includes a false premise: the assumption that the existence of a unitary operator connecting the folia of two representations implies a unitary equivalence between those representations.  As we have shown, though, there is no such implication if the unitary operator is what we've called a Weyl unitary.  Arageorgis's argument is unsound.

This means that at least one significant part of the empirical content of a quantum theory -- its transition probabilities -- can be preserved by a mapping between the folia of two inequivalent representations.  If we further assume (as conventional wisdom dictates) that a quantum theory's symmetries preserve all empirical content, then the folia of at least some pairs of inequivalent representations must be empirically equivalent if spontaneous symmetry breaking is possible.  The notion that unitary equivalence is a necessary condition for physical equivalence should now appear quite suspect.  Insofar as the so-called ``Hilbert space conservative'' interpretation of quantum field theory identifies physical equivalence with unitary equivalence \citep[see][]{LRinterpreting}, that interpretation must come into question as well.

\section*{Appendix 1: Representation Wigner Theorem}

Here we prove the results mentioned in the main text.

\begin{defn} Let $\om$ be a state on a $C^*$-algebra $\2A$, let $\8H$
  be a Hilbert space with $\Om$ a unit vector in $\8H$, and $\pi
  :\2A\to B(\8H )$ a representation of $\2A$.  We say that the triple
  $\langle \8H,\pi ,\Om \rangle$ is a \emph{GNS representation} for
  $\om$ just in case:\begin{enumerate}
  \item $\langle \Om ,\pi (A)\Om\rangle = \om (A)$, for all $A\in
    \2A$, and
  \item $\{ \pi (A)\Om :A\in \2A \}$ is dense in the Hilbert space
    $\8H$. \end{enumerate} \end{defn}

The GNS theorem shows that for each state $\om$, there is a GNS
representation; and that any two GNS representations of $\om$ are
unitarily equivalent.

Let $\2A$ be a $C^*$-algebra, let $\om$ be a state of $\2A$, and let
$\a$ be a $*$-automorphism of $\3A$.  Let $(\3H,\pi ,\Om)$ be the GNS
triple of $\2A$ induced by $\om$, and let $(\3H ',\pi ',\Om ')$ be the
GNS triple of $\2A$ induced by $\om\circ\a^{-1}$.  For brevity, we
sometimes just use $\pi$ and $\pi '$ to denote the corresponding
triples.

\begin{gwt} Let $\langle \8H,\pi,\Om\rangle$ be a GNS representation
  for $\om$, and let $\langle \8H',\pi',\Om '\rangle$ be a GNS
  representation for $\om\circ\alpha^{-1}$.  Then there is a unique
  unitary operator $W=W_{\pi ,\pi'}:\8H\to \8H'$ such that $W\Om =\Om
  '$, and $W\pi (\alpha ^{-1}(A))= \pi'(A)W$ for all $A\in \2A$.
\end{gwt}

\begin{proof} Let $(\8H ,\pi ,\Om)$ be a GNS representation of $\2A$
  for the state $\om$, and let $(\8H ',\pi ',\Om ')$ be a GNS
  representation of $\2A$ for the state $\om\circ \al ^{-1}$.  Define
  $W:\8H\to\8H '$ by setting
$$ W\pi (A)\Om = \pi '(\al (A))\Om ' ,\qquad \forall A\in
\2A.$$ Since $\al (I)=I$, it follows that $W\Om =\Om '$.  Since
$$ \norm{\pi '(\al (A))\Om '}^2 = \langle \Om ',\pi '(\al (A^*A))\Om
'\rangle = \om (\al ^{-1}(\al (A^*A)))=\om (A^*A)=\norm{\pi (A)\Om
}^2,$$ it follows that $W$ is well defined and extends uniquely to a
unitary operator from $\8H$ to $\8H '$.  Note that since $\pi (A)\Om
=W^*W\pi (A)\Om =W^*\pi '(\al (A))\Om '$, it follows that $W^*\pi
'(B)\Om '=\pi (\al ^{-1}(B))\Om$ for all $B\in \2A$.  Therefore,
$$ W^*\pi '(A)W\pi (B)\Om = W^*\pi '(A\al (B))\Om '= \pi (\al
^{-1}(A)B)\Om = \pi (\al ^{-1}(A))\pi (B)\Om ,$$ for all $A,B\in \2A$.
Since the vectors $\pi (B)\Om$, for $B\in \2A$, are dense in $\8H$, it
follows that $W^*\pi '(A)W=\pi (\al ^{-1}(A))$ for all $A\in \2A$.
That is, $W$ implements a unitary equivalence from $\pi \circ \al
^{-1}$ to $\pi '$.

To show the uniqueness of $W$, it suffices to note that $\Om$ is a
cyclic vector for $\pi\circ\al$, $\Om '$ is a cyclic vector for $\pi
'$, and $W\Om =\Om '$.  Thus, there is at most one unitary intertwiner
from $\pi\circ\al ^{-1}$ to $\pi'$ that maps $\Om$ to
$\Om'$. \end{proof}

Note that if $\alpha =\iota$ is the identity automorphism, and if we
take $\langle \8H',\pi',\Om '\rangle = \langle \8H,\pi,\Om \rangle$,
then $I$ satisfies the conditions of the theorem, hence by uniqueness
$W_{\pi,\pi'}=I$.

For the following corollary, recall that any representation $\pi :\2A
\to B(\8H )$ of a $C^*$-algebra gives rise to a map $\pi ^*$ of unit
vectors of $\8H$ into the state space of $\2A$.  In particular,
$$\pi ^*(x)(A) = \langle x,\pi (A)x\rangle ,\qquad (A\in \2A) .$$

\begin{cor} Let $(\8H ,\pi ,\Om)$ be a GNS representation for $\om$.
  Then the Wigner unitary $W$ for $\al$ implements the action of $\al$
  on vectors in $\8H$.  That is, $(\pi ')^*(Wx)=\pi ^*(x)\circ \al
  ^{-1}$ for any unit vector $x$ in $\8H$.  \end{cor}

\begin{proof} By the Theorem, a Wigner unitary $W$ intertwines $\pi\circ
  \al ^{-1}$ and $\pi'$, that is
$$ \pi (\al ^{-1}(A))= W^*\pi '(A)W ,$$
for all $A\in \2A$.  Hence $$ (\pi ')^*(Wx)(A) \:=\: \langle Wx,\pi '
(A)Wx \rangle \: =\:\langle x,W^*\pi '(A)Wx\rangle \:=\: \langle
x,\pi (\al ^{-1}(A))x\rangle\:=\: \pi ^*(x)(\al ^{-1}(A)) ,$$ for all
$A\in \2A$.
\end{proof}

The preceding corollary can be conveniently pictured via a commuting diagram:

$$\bfig
\square/>`<-`<-`>/[S(\2A)`S(\2A)`\8H`\8H ';\al '`\pi ^*`(\pi ')^*`W]
\efig $$ where $S(\2A )$ is the state space of $\2A$, and $\al
':S(\2A)\to S(\2A )$ is the symmetry $\om\mapsto \om \circ \al ^{-1}$.

\begin{cor2} If the Wigner unitary $W:\8H\to\8H'$ also induces a
  unitary equivalence between $\pi$ and $\pi'$, then $\al '\circ \pi
  ^*=\pi ^*$.  \label{trivialize} \end{cor2}

Recall that $\al ':S(\2A )\to S(\2A )$ is defined by $\al '(\om )=\om
\circ \al ^{-1}$.

\begin{proof} If $W$ induces a unitary equivalence between $\pi$ and
  $\pi '$ then
$$ \pi (A)W^* = W^*\pi '(A) = \pi (\al ^{-1}(A))W^* ,$$
for all $A\in \2A$.  Canceling the unitary operator $W^*$ on the right
gives $\pi (A)=\pi (\al ^{-1}(A))$ for all $A\in \2A$, that is $\pi =
\pi \circ \al ^{-1}$.  From the latter equation it clearly follows
that $\pi ^*=\al '\circ \pi ^*$. \end{proof}

\section*{Appendix 2: Properties of the Wigner unitary operator}

We now give a special case of the Representation Wigner Theorem which will illustrate some properties of the Wigner unitary.  But first we need a lemma.

\begin{lemma} If $\langle \8H,\pi ,\Om \rangle$ is a GNS
  representation for the state $\om$ then $\langle \8H,\pi \circ
  \alpha ^{-1},\Om \rangle$ is a GNS representation for the state $\om
  \circ \al ^{-1}$.
\end{lemma}

\begin{proof} Since $\langle \Om ,\pi (\al ^{-1}(A))\Om \rangle = \om
  (\al ^{-1}(A))$ and since $\Om$ is cyclic under $\{ \pi (\al
  ^{-1}(A)) :A\in \mathfrak{A}\}$, it follows that $\langle \8H,\pi
  \circ \al ^{-1},\Om \rangle$ is a GNS representation for $\om \circ
  \al ^{-1}$. \end{proof}

We can now apply the Representation Wigner Theorem to the representations
$\langle \8H,\pi,\Om\rangle$ and $\langle \8H,\pi\circ \alpha
^{-1},\Om\rangle$

\begin{swt} If $W:\8H\to \8H$ is the Wigner operator for the
  representations $\langle \8H,\pi,\Om \rangle$ and $\langle
  \8H,\pi\circ \alpha ^{-1},\Om \rangle$, then $W=I$.
\end{swt}

\begin{proof} By RWT, $W\Om =\Om$ and $W\pi (\al ^{-1}(A))=\pi (\al
  ^{-1}(A))W$ for all $A\in \2A$.  Since $\al$ is an automorphism,
  $W\pi (A)=\pi (A)W$ for all $A\in \2A$.  Hence
$$ W\pi (A)\Om = \pi (A)W\Om = \pi (A)\Om ,$$
for all $A\in \2A$.  Since the set $\{ \pi (A)\Om :A\in
\mathfrak{A}\}$ of vectors is dense in $\8H$, it follows that $Wx=x$
for every vector in $\8H$; that is, $W=I$. \end{proof}

We now apply SRWT to the case of symmetries in elementary quantum mechanics.  Let $\8H$ be a finite-dimensional Hilbert space, and let $U:\8H\to \8H$ be
a unitary operator that induces a symmetry.  Then we have the
following transformations:
\[ \begin{array}{ll}
  \varphi \longmapsto U\varphi & \text{transformed state} \\
  A\longmapsto UAU^* & \text{transformed observable} \end{array} \] Of
course, $B(\8H )$ is a $C^*$-algebra, and each (pure) state of $B(\8H
)$ is represented uniquely by a ray in $\8H$.  The unitary $U$ induces
the automorphism $\alpha (A)=UAU^*$ of $B(\8H )$, as well as the
corresponding state mapping.

In order to apply SRWT, we need to find representations.  The first
representation is $\langle \8H,\iota ,\varphi\rangle$, where $\iota
:B(\8H )\to B(\8H )$ is the identity, and $\varphi$ is an arbitrarily
chosen unit vector.  The second representation is $\langle \8H ,\iota
\circ \al ^{-1},\varphi\rangle$.  By GWT, there is a Wigner unitary
$W:\8H\to \8H$, and by SRWT, $W=I$.

So, in what sense does $W$ induce the symmetry $U$ on states?  Should
we not have $W=U$?  No, because a vector $\varphi$ in $\8H$ names
different states on $B(\8H )$ according to which representation we
consider, either $\iota$ or $\alpha ^{-1}$.  Relative to the first,
$\varphi$ represents the state $A\mapsto \langle \varphi
,A\varphi\rangle$, and relative to the second, $\varphi$ represents
the state $A\mapsto \langle \varphi ,\al ^{-1}(A)\varphi \rangle$.

What $W$ does is to map a vector representing some state $\om$
relative to $\pi$ to a vector representing the state $\om \circ \al
^{-1}$ relative to $\pi\circ \al ^{-1}$.  In the way we have set
things up, $W=1_H$, which just means that if $\varphi$ represents
$\om$ relative to $\pi$, then $\varphi$ represents $\om\circ\al ^{-1}$
relative to $\pi'=\pi\circ \al ^{-1}$.  So, indeed, the identity map
implements the symmetry $\om \mapsto \om \circ \al^{-1}$ of states!

Let us look now, more generally, at the case of an
\emph{unbroken} symmetry.  By hypothesis, the symmetry $\alpha$ is unbroken
just in case the representations $(\8H ,\pi ,\Om )$ and $(\8H
,\pi\circ \al ^{-1},\Om )$ are unitarily equivalent.  That is, there
is a unitary operator $V:H\to H$ such that $V\pi (\al ^{-1}(A)) = \pi
(A)V$.  In fact, in the most interesting case where $\om$ is a pure
state, $V$ can be chosen such that $V=\pi (U)$ for some unitary
operator $U\in \2A$, hence
$$ \pi (\al (A)) \: = \: V\pi (A)V^* \:=\: \pi (UAU^*) ,$$
for all $A\in \2A$.  (To verify the existence of such a $U\in \2A$,
see \citet[730]{kr}.)

Of course, we are still guaranteed the existence of the Wigner Unitary
$W:\8H\to \8H$.  (In fact, we know that $W=I$; but ignore that fact
for now.)  Which operator, $W$ or $V$, implements the symmetry
$\alpha$ on states?  The answer is that they \emph{both} do, but in
different senses.

Compare the following two diagrams:
$$\bfig
\square/>`<-`<-`>/[S(\2A)`S(\2A)`\8H`\8H ;\al '`\pi ^*`(\pi ')^*`W]
\square(1000,0)/>`<-`<-`>/[S(\2A)`S(\2A)`\8H`\8H;\al '`\pi ^*`\pi
^*`V] \efig $$

The square on the left shows the action of the Wigner
unitary $W$ for the special case of the GNS representation $(\8H
,\pi\circ \al ^{-1},\Om)$ for $\om\circ \al ^{-1}$.  The square on the
right shows that action of the unitary $V$ that implements the
equivalence between $\pi$ and $\pi\circ \al ^{-1}$.  The key
difference, of course, is that $V$ implements the symmetry in such a
way that the correspondence between vectors and states can be held
invariant (the vertical arrows are the same), whereas $W$'s
implementation requires a change of correspondence ($\pi ^*$ versus
$(\pi ')^*$).  But a state is a way to map observables to numbers, so
changing the correspondence between vectors and states is equivalent
to leaving this correspondence fixed and instead changing the labels
of observables.  In equation form:
$$ (\pi ')^* = \al '\circ \pi ^* ,$$
i.e.\ the correspondence $(\pi ')^*$ matches vectors with an
observable $A$ in exactly the way that the correspondence $\pi ^*$
matches vectors with the observable $\al ^{-1}(A)$.

\subsection*{Acknowledgements}
Thanks to Wayne Myrvold and Giovanni Valente for detailed comments on
an earlier draft.  DJB would like to thank Laura Ruetsche for several
valuable discussions of the seeming paradox and John Earman for
helpful correspondence.  HPH thanks Jeremy Butterfield, Joe Rachiele,
and Noel Swanson for helping him clarify the contrast between Wigner
and intertwining unitaries.

\bibliography{DAVEbib}

\end{document}